\documentclass[conference]{IEEEtran}
\usepackage{amsmath}
\usepackage{amsthm}
\usepackage{graphicx}
\usepackage{hyperref}
\usepackage{cleveref}

\theoremstyle{plain}
\newtheorem{theorem}{Theorem} 
\newtheorem{lemma}{Lemma}

\theoremstyle{definition}
\newtheorem{definition}{Definition}

\newcommand{\hxy}{\text{H}(X|Y)}

\newcommand{\cF}{{\cal F}}
\newcommand{\eps}{\epsilon}

\newcommand{\Ceil}[1]{\ensuremath{\left\lceil#1\right\rceil}}

\long\def\symbolfootnote[#1]#2{\begingroup%
\def\thefootnote{\fnsymbol{footnote}}\footnotetext[#1]{#2}\endgroup}

\begin{document}

\title{Network Coded Gossip with Correlated Data}

%\iffalse
\author{\IEEEauthorblockN{Bernhard Haeupler\IEEEauthorrefmark{1}}
\IEEEauthorblockA{%CSAIL, RLE\\
MIT\\
haeupler@mit.edu}
\and
\IEEEauthorblockN{Asaf Cohen\IEEEauthorrefmark{1}}
\IEEEauthorblockA{%Department of Communication Systems Engineering\\
Ben-Gurion University\\% of the Negev, Israel\\
coasaf@cse.bgu.ac.il}
\and
\IEEEauthorblockN{Chen Avin}
\IEEEauthorblockA{%Department of Communication Systems Engineering\\
Ben-Gurion University\\% of the Negev, Israel\\
avin@cse.bgu.ac.il}
\and
\IEEEauthorblockN{Muriel M\'edard}
\IEEEauthorblockA{%RLE\\
MIT\\
medard@mit.edu}}

\maketitle

%create a star-footnote pointing out the equal contribution using a larger-than-standard fontsize.
\symbolfootnote[1]{\small A. Cohen and B. Haeupler contributed equally to this work.}

\begin{abstract}
We design and analyze gossip algorithms for networks with correlated data. In these networks, either the data to be distributed, the data already available at the nodes, or both, are correlated. % This model is applicable for a variety of modern networks, such as sensor, peer-to-peer and content distribution networks. 
Although coding schemes for correlated data have been studied extensively, the focus has been on characterizing the rate region in static memory-free networks. In a gossip-based scheme, however, nodes communicate among each other by continuously exchanging packets according to some underlying communication model. The main figure of merit in this setting is the stopping time -- the time required until nodes can successfully decode. While Gossip schemes are practical, distributed and scalable, they have only been studied for uncorrelated data. 

We wish to close this gap by providing techniques to analyze network coded gossip in (dynamic) networks with correlated data. We give a clean framework for oblivious network models that applies to a multitude of network and communication scenarios, specify a general setting for distributed correlated data, and give tight bounds on the stopping times of network coded protocols in this wide range of scenarios. 
\end{abstract}

\IEEEpeerreviewmaketitle
%%%%%%%%%%%%%%%%%%%%%%%%%%%%%%%%%%%%%%%%%%%%%%%%%%%%%%%%%%%%%%%%%%%%%%%%%%%%%%%%%
\section{Introduction}
In this paper, we design and analyze information dissemination algorithms in communication networks with correlated data. In these networks, either the data to be distributed, the data already available at the nodes, or both, are correlated. This problem arises in a many networking applications, such as sensor, peer-to-peer or content distribution networks. One such example is a large set of distributed temperature sensors with a clock at the receiver. Both the temperatures at different sensor locations and the time at which a measurement is taken have high correlations among each other. 

%Networks with correlated data are networks where either the data already available at the nodes (before the transmission) is correlated with the data to be sent, or where the data to be sent is distributed among the nodes as several correlated sources. 

While the current information theory literature includes several coding schemes for correlated data, the focus in these works is mainly on characterizing the rate region - the set of achievable rates. On the other hand, recent work in the networking literature offers a multitude of efficient, decentralized and address-oblivious schemes for information dissemination (e.g., randomized gossip). Unfortunately these schemes treat the data as uncorrelated and neglect any available information at the receivers. 
The focus of this paper is to close this gap and give tools for analyzing gossip-based algorithms in networks with correlated data. 

%%%%%%%%%%%%%%%%%%%%%%%%%%
\subsection{Related Work}

{\em Distributed Source Coding:}\  Distributed compression has been studied in information theory mainly through small, canonical problems. In \cite{SlepianWolf73}, Slepian and Wolf considered the problem of separately encoding two correlated sources and joint decoding. In \cite{AhlswedeKorner75} and \cite{CohenAvestEffros_ISIT09} the problem of compression with a rate-limited helper is considered. In \cite{Ho_et_al06}, Ho et al.\ considered the multicast problem with correlated sources which can be viewed as extending the Slepian-Wolf problem to arbitrary networks through \emph{network coding}. Further extensions appeared in \cite{Barros_Servetto06} and \cite{BakshiEffros08}. In all these studies the goal is to characterize the \emph{rate region} for (fixed) static and memory-free networks, that is, the set of required capacities needed for a multicast. 
 
\vspace{0.1cm}
{\em (Network Coded) Gossip:}\  Gossip schemes were first introduced in \cite{demers1987epidemic} as a simple and decentralized way to disseminate a piece of information in a network. A detailed analysis of a class of these algorithms is given in \cite{karp2000randomized}. In these schemes nodes communicate by continuously picking communication partners in a randomized fashion and then forwarding the information. The main figure of merit is the stopping time -- the time needed for all nodes to be informed. Such randomized gossip-based protocols are attractive due to their locality, simplicity, and structure-free nature, and have been offered in the literature for various tasks. For the task of distributing multiple messages \cite{Deb_Medard_Choute06} introduced \emph{algebraic gossip}, a network coding-based gossip protocol in which nodes exchange linear combinations of their available messages. This idea was extended to arbitrary networks in \cite{Mosk-Aoyamam_Shah06} and \cite{Borokhovich2010Tight}. Haeupler~\cite{haeupler2011analyzing} proved tight bounds for the stopping time of algebraic gossip for various models, including (adversarial) dynamically changing networks~\cite{NCdynamicnetworks11} and nodes with limited memory~\cite{NCgossipwithlimitedmemory}. Improved bounds for non-uniform gossip were given in \cite{DBLP:conf/podc/AvinBCL11}. The projection analysis developed in \cite{haeupler2011analyzing} will play a key role in this paper.

\subsection{Our Contributions}

To our knowledge this paper is the first to combine these two strains of research and analyze gossip based protocols in networks with correlated data. Our contributions are manifold: First, we give a clean and general framework for oblivious network models in \Cref{sec:obliviousnetworks} and define a setting for a correlated environment in \Cref{sec. model}. In this general setting we extend the projection analysis of \cite{haeupler2011analyzing} by making a connection between the coefficient vectors a node knows and the amount and type of information it has learned. This results in simple, direct and self-contained proofs of tight bounds on the stopping time in the canonical models of one source and side information at the receivers, as well as two correlated sources. In \Cref{sec:capacities} we then give results for the general scenario of multiple sources and side information. We do this by providing tight bounds on the time required for any set of (fractional) capacities to be induced by the (random) packet exchanges generated in an oblivious network model. This allows to transform results on the rate region of static memory-free networks (e.g., \cite{Ho_et_al06,BakshiEffros08}) into bounds on the stopping times of gossip-based algorithms. These capacity bounds are interesting in their own right and have the potential to be useful in other information dissemination problems. %They could also be used to prove the fault-tolerance and robustness of gossip networks based on establishing bounds on their connectivities.  
%%%%%%%%%%%%%%%%%%%%%%%%%%%%%%%%%%%%%%%%%%%%%%%%%%%%%%%%%%%%%%%%%%%%%%%%%%%%%%%%%

\section{Network and Information Model}\label{sec. model}
In this section we state the broadcast problem, define the communication model and the information model which  and state the information model defining the nature of the
source and side information. 

\subsubsection{Network and Communication Model}

For simplicity, we will assume that the network consists of a fixed set $V$ of $n$ nodes. Communication takes place in synchronous rounds. In each round, each node $v$ decides on a packet $p_v$ of $s$ bits to be sent out (possibly using randomness). Given the current state of the network, the network model then specifies (possibly using randomness) which packet will be received by which node in the current round. This corresponds to a probability distribution over directed edge sets where a directed edge $(v,u)$ means that the packet $p_v$ is successfully received by node $u$. Nodes are assumed to have unlimited storage (for schemes with limited buffers see~\cite{NCgossipwithlimitedmemory}). We denote the set of directed edges chosen for round $t$ with $E_t$ and call it the active edge set for time $t$. 

\subsubsection{Source and Side Information}

We assume $k$ messages are initially distributed over the network. The $i$-th message constitutes of $l$ i.i.d.\ samples from the random variable $X_i$, namely, a vector $x_i$. The message vectors $x_1,\ldots,x_k$ are initially distributed to nodes (i.e., sources) such that each vector is given to at least one node. We also assume that each node $v \in V$ (or terminal) has some side information. The side information $y_v$ at node $v \in V$ is drawn as $l$ i.i.d. samples from $Y_v$. Note that the variables $\{X_i\}_{i=1}^{k} \cup \{Y_v\}_{v \in V}$ are arbitrarily correlated according to some known memoryless joint distribution.
We are interested in the time when nodes are able to decode $x_1,\ldots,x_k$ based on their side information and the packets exchanged with other nodes. 

%In general, we assume that the probabilities $p(x_1, \ldots, x_k,y_v)$ are known to all nodes. However, note that in most coding schemes, fewer information is required. For example, in random binning encoding and joint typicality decoding, the sources use only the conditional entropies, while the terminals use the joint distribution. If minimum entropy decoding is used, then receivers need no knowledge about the probabilities~\cite{Ho04networkcoding}, making the coding scheme universal~\cite{csiszar1982linear}. 

%%%%%%%%%%%%%%%%%%%%%%%%%%%%%%%%%%%%%
\subsubsection{The Encoding and Decoding Schemes}

For a given field size $q$ and slack $\delta>0$ we assume throughout that nodes employ the following coding scheme:
Prior to communication, source nodes perform random binning, that is, for every $1 \leq i \leq k$ each node receiving the message vector $x_i$ applies the same random mapping into $2^{l(H(X_i)+\delta)}$ bins. 
%Comment Bernhard: I uncommented these remarks since binning with a rate of H(X) is essential for the proof of Theorem 3.
%Note that in certain cases binning at a higher or lower rate is beneficial. Assume now binning into $2^{2l\log q}$ bins. 
The resulting bin indices (which are the same for every node initialized with $x_i$) are interpreted as vectors of length $h = \Ceil{\frac{l}{\log q}(H(X_i)+\delta)}$ over the finite field $\cF_q$. These vectors are then split into $\frac{h \log q}{s}$ blocks of $\frac{s}{\log q}$ symbols in $\cF_q$ each, for a total of $s$ bits per block.

During the communication phase nodes sends out random linear combinations (over $F_q$) of these blocks as packets. To keep track of the linear combination contained in a packet one coefficient for each block of each message is introduced and sent in the header of each packet. As in all prior works on distributed network coding (e.g., \cite{Deb_Medard_Choute06,Mosk-Aoyamam_Shah06,Borokhovich2010Tight,haeupler2011analyzing,NCgossipwithlimitedmemory,DBLP:conf/podc/AvinBCL11}), we assume that $\frac{s}{\log q}$ is sufficiently large compared to the number of coefficients. This renders the overhead of the header negligible leaving a packet size of $s$ bits as desired. 

At each node, independent linear equations on the blocks are collected for decoding. We denote with $S_v$ the subspace spanned by all coefficient vectors received at node $v$. We also use the following notion of knowledge from \cite{haeupler2011analyzing}:

\begin{definition}\label{def:knowledge} A node knows a coefficient vector $\mu$ iff $S_v$ is not orthogonal to $\mu$, that is, there exists a vector $c \in S_v$ such that $\langle c,\mu \rangle \ne 0$. 
\end{definition}

Lastly, we will make use of the following lemma on random binning. 

\begin{lemma}[\cite{BakshiEffros08,CohenAvestEffros_ISIT09}]\label{lem. rank only}
Let $X \in \mathcal{X}$ and $Y \in \mathcal{Y}$ be two arbitrarily correlated random variables and let $x,y$ be two vectors that are created by taking $l$ i.i.d. samples from their joint distribution. Suppose, for some $\eps>0$, all possible sequences in $\mathcal{X}^l$ are randomly and uniformly divided to at least $2^{l(\hxy +\delta)}$ bins. Then joint typicality decoding correctly decodes $x$ with high probability (as $l \to \infty$) using $y$ and any $\Ceil{(H(X|Y)+\delta)l}$ bits of information on the bin index of the bin in which the true $x$ resides. 
\end{lemma}
%For completeness, a proof sketch of the random binning scheme is given in Appendix \ref{app. random binning}.
In particular, Lemma \ref{lem. rank only} asserts that having access to the side information vector $y$, the message vector $x$ can be decoded with high probability using \emph{any} $\Ceil{\frac{l}{s}(H(X|Y)+\delta)}$ \emph{linearly independent equations} on the blocks describing the bin index of $x$.% This will play a critical role in our analysis. 
%%%%%%%%%%%%%%%%%%%%%%%%%%%%%%%%%%%%%

\section{Oblivious Network Models}\label{sec:obliviousnetworks}

In this section we introduce the definition of an oblivious network model. This gives a clean and very general framework capturing a wide variety of communication and (dynamic) network settings. While this was already somewhat implicit in \cite{haeupler2011analyzing} we restrict ourself to networks without adaptive adversarial behavior. This greatly facilitates the much cleaner framework presented in this section. 

\begin{definition}
A network model is \emph{oblivious} if the active edge set $E_t$ of time $t$ only depends on $t$, $E_t'$ for any $t' < t$ and some randomness. We call an oblivious network model furthermore i.i.d. if the active edge set $E_t$ is sampled independently for every time $t$ from a distribution of edge sets. 
\end{definition}

The following are examples of oblivious (and i.i.d.) network models:
In the {\em (Uniform) Gossip Model}~\cite{Mosk-Aoyamam_Shah06,Borokhovich2010Tight,haeupler2011analyzing} one has an underlying (directed) graph $G$ and in each round each node picks a random neighbor as a communication partner. A node then sends (PUSH) or requests (PULL) a packet from its partner or both (EXCHANGE). The {\em Rumor Mongering}~\cite{Deb_Medard_Choute06} or {\em Random Phone Calls Model}~\cite{karp2000randomized} is a well-studied special case of this in which $G$ is complete, that is, nodes pick a random node as a communication partner. It is easy to include more sophisticated features in an oblivious network model. Random packet losses in wired networks, or characteristics of radio networks like half-duplex transmission, collisions, packet loss rates depending on SNR and more can be easily modeled by removing edges according to (randomized) rules. 
%One example would be to model collisions which allow each node to only receive at most one message per round, namely if it is listening and exactly one neighbor sends. 
%\commentb{Switch to graph representation here}
%This behavior can be added to any network model by removing all occurrences of a node $r$ in a hyperedge $(s,\{r,r_1,r_2,\ldots\})$ if there are at least two hyperedges containing $r$ as a receiver. 
An interesting example of non-i.i.d. oblivious network models are {\em (edge-)markovian evolving graphs}~\cite{clementi2008flooding}.%,avin2008explore,clementi2011information}:\\

%The \emph{flooding time} defined next is an important measure for an oblivious network model. As demonstrated below it has deep connections to how fast information can propagate through a network. 
%dithe speed of information 
%An important measurement for how well an oblivious network model mixes is its \emph{flooding time}. As we will see, there is a strong relationship between the the speed of information dissemination and the flooding time of a network. 
%This does not seem to be needed anymore
%For a (directed) graph $G = (V,E)$ let $\Gamma^+_{G}(S)$ be the inclusive neighborhood of a subset of nodes $S \subseteq V$, that is, 
%$\Gamma^+_{G}(S)  = \{v \in V | \exists s \in S: \{s,v\} \in E \vee v = s\}$.

For any oblivious network model $M$ we can define a random flooding process $F(M,p,S)$. Informally, this process describes which nodes are informed over time if initially only nodes in $S$ are informed and from there on every informed node informs all its communication partners (as specified by $M$). The only modification to this standard infectious process is the parameter $p$ which adds an independent probability of $p$ for each transmission to be overheard. 

\begin{definition}\label{def:oblivious}
Let $M$ be an oblivious network model, $p$ be a probability of fault and $S \subseteq V$ be a starting set of nodes.
We define the flooding process $F(M,p,S)$ to be the random process $S_1 \subseteq S_2 \subseteq \ldots$ that
is characterized by $S_1 = S$ and for every time $t$ we define $S_{t+1}$ by taking each of the (directed) edges $E_t$
specified by $M$ for time $t$ independently with probability $1 - p$ to obtain $E'_t$ and then we set $S_{t+1} = \{v \in V\ |\ \exists u \in S_{t}:\ (u,v) \in E'_t \ \vee \ v = u\}$.
\end{definition}

Note, that \Cref{def:oblivious} is only well defined if $M$ is oblivious. Furthermore, $F$ becomes a time-homogeneous Markov chain if $M$ is an i.i.d. oblivious network model. Also, as long as for every time $t$ the union over the edges in $M$ from $t$ to infinity is almost surely connected then $F$ is absorbing with the unique absorbing state $V$. We say the flooding process $F$ stops if it reaches this absorbing state and we denote the time this happens with the random variable $S_F$. The next definition pairs this flooding time with a throughput parameter $\alpha$ that corresponds to the exponent of the flooding process tail probability. The reason for this definition and its connection to the multi-message throughput behavior of network coding becomes apparent in the statement and proof of \Cref{thm:oldthm} below. 

\begin{definition}\label{def:floodingtime and throughput}
We say an oblivious network model $M$ on a node set $V$ floods in time $T$ with throughput $\alpha$ if there exists a prime power $q$ such that for every $v \in V$ and every $k > 0$ we have $P[S_{F(M,1/q,\{v\})} \geq T + k] < q^{-\alpha k}$.
\end{definition}

To give a few illustrating examples of flooding times we note that the random phone call network model on $n$ nodes floods in $\Theta(\log n)$ time with constant throughput. The uniform gossip model on a connected degree bounded graph $G$ floods in time $\Theta(D)$ and with constant throughput where $D$ is the diameter of $G$. In many oblivious network models it is easy to give tight bounds on the flooding time and throughput. 

%\begin{definition}
%We say a random variable $X$ is concentrated with exponent $\alpha$ 
%beyond $T$ if for every $k > 0$ we have $P[X \geq T + k] < 2^{-\alpha k}$.
%\end{definition}

%Note that for most oblivious and time-homogeneous network models $M$ the stopping time 
%of the flooding process $F = F(M,p,\{v\})$, where $v \in V$ is any node, is concentrated
%beyond $C \cdot E[S_F]$ for some constant $C$. % with exponent $\alpha = \log O(p^{-1})$.

%In particular for any constant $p$ the flooding time of the random phone call network model
%on $n$ nodes is concentrated beyond $\Theta(\log n)$ with constant exponent. The flooding 
%time of the uniform gossip model with constant $p$ on a graph $G$ with $n$ nodes is also
%concentrated beyond $\Theta(D)$ where $D$ is the diameter. The exponent in this case is the following min-cut quantity:
%$$\alpha = \Theta\left(\min_{\emptyset \neq S \subset V} \sum_{(u,v) \in G \cap (S \times \overline{S})} 1/d(u) + 1/d(v)\right) \geq \Omega(1/\Delta)$$
%where $d(u)$ is the degree of node $u$ and $\Delta = \max_u d(u)$ is the maximum degree~\cite{haeupler2011analyzing}.

%\subsubsection{Uncorrelated Information Spreading}

With this framework for oblivious network models in place we can give a cleaner restatement of Theorem 3 in \cite{haeupler2011analyzing}. We also provide a sketch of the proof since we will later expand on the ideas used therein. 

\begin{theorem}[Theorem 3 of \cite{haeupler2011analyzing}]\label{thm:oldthm}
Suppose $M$ is an oblivious network model that floods in time $T$ with throughput $\alpha$. 
Then, for any $k$, random linear network coding in the network model $M$ spreads $k$ arbitrarily distributed messages to all nodes with probability $1 - \eps$ after $T' = T + \frac{1}{\alpha}(k + \log \eps^{-1})$ rounds.
\end{theorem}
\begin{proof}
The random linear network coding protocol we analyze will use the same field size that achieves the parameters $T$ and $\alpha$ for $M$ in Definition~\ref{def:floodingtime and throughput}. We fix a coefficient vector $\mu \in F_q^k$. This vector is initially known to a non-empty subset $S$ of 
nodes. It is easy to check that the probability that a node $v$ does not know $\mu$ after it has received a packet from a node that knows $\mu$ is at most $1/q$. This implies that knowledge of $\mu$ spreads through the network exactly as the flooding process $F_{M,1/q,S}$. Using the assumption, Definition~\ref{def:floodingtime and throughput} and the monotonicity of $S_{F(M,1/q,S)}$ in $S$ we get that the probability that a vector $\mu \in F_q^k$
is not known to all vectors after $T'$ steps is at most $q^{-(k + \log \eps^{-1})}$. A union bound over all $q^k$ vectors shows that with probability at least $1 - q^{-\log \eps^{-1}} \geq 1 - \eps$ all node will know about all vectors and it is easy to see that this implies that all nodes are able to decode the messages. 
\end{proof}

%%%%%%%%%%%%%%%%%%%%%%%%%%%%%%%%%%%%%%%%%%%%%%%%%%%%%%%%%%%%%%%%%%%%%%%%%%%%%%%%%

\section{Simple, Direct Proofs for Tight Stopping Times}
In this section we give a simple, direct derivation of tight stopping time bounds for gossip with one source and side information at the nodes and gossip with two correlated sources. Our two main results in this section are the following.
\begin{theorem}\label{theo. one source}
Suppose $M$ is an oblivious network model that floods in time $T$ with throughput $\alpha$. We assume a single message $x$ generated from $X$ and side information $y_v$ generated from $Y_v$ at every node. Fix an error probability $\epsilon>0$. Then, for any $\delta>0$, for any large enough block length $l$ and any packet size $s$, node $v$ will correctly decode $x$ with probability at least $1-\epsilon$ after $T' = T+\frac{1}{\alpha}\left(\frac{l}{s}(H(X|Y_v)+\delta)+\log \eps^{-1} + 3 \right)$ rounds. 
\end{theorem}

\begin{theorem}\label{thm: two sources no sideinformation} 
Suppose $M$ is an oblivious network model that floods in time $T$ with throughput $\alpha$. We assume two messages $x_1,x_2$ are generated from $X_1,X_2$ and nodes have no side information. Fix an error probability $\epsilon>0$. Then, for any $\delta>0$ and for large enough $l$, with probability at least $1-\epsilon$ every node will correctly decode $x_1,x_2$ after $T+\frac{1}{\alpha}\left(\Ceil{\frac{l}{s}(H(X_1,X_2)+2\delta)}+\log \eps^{-1} + 3\right)$ rounds.
\end{theorem}

The idea for proving these theorems is to generalize the observation of \cite{haeupler2011analyzing} that the question of when a node can decode is equivalent to determining when this node knows (see \Cref{def:knowledge}) enough coefficient vectors. The proof of \Cref{thm:oldthm} shows that flooding or spreading of knowledge of vectors can be analyzed using a union bound. This implies that only the number of vectors needed is of importance. In the case with \emph{uncorrelated} sources and no side information essentially knowledge of all coefficient vectors is necessary. In the correlated scenario, however, we want to relate the number of vectors a node $v$ needs to know to the conditional entropy $H(X|Y_v)$. \Cref{lem. rank only} helps in this respect. It asserts that in order to decode, a node with side information $Y$ does not necessarily need $i = \Ceil{(H(X|Y)+\delta)l}$ specific bits, but rather, assuming joint typicality decoding, it requires only \emph{any} sufficient amount of information about the index of the bin in which $x$ resides. This is achieved by {\em any} $i/s$ packets containing independent equations on the bin index. We can thus focus on the knowledge a node is required to obtain in order for its coefficient matrix to have rank at least $i/s$. 

Unfortunately it is possible that a node knows many vectors without having a large rank. In fact, upon reception of the first packet a node gets to know all but a $1/q$ fraction of all vectors. On the other hand, in order to prove faster stopping times we want to argue that the knowledge of only an exponentially small fraction of all vectors suffices for decoding. This is achieve by the following lemma which shows that indeed only $q^l$ {\em specific coefficient vectors} suffice to guarantee that at least $l$ independent coefficient vectors were received:

\begin{lemma}\label{lemma q^l+1}
Let $\mathcal V$ be be a finite dimensional vector space over a finite field $\cF_q$. For every $0 \leq h < \dim \mathcal V$ there exist $w = q^h + 1$ vectors $v_1, \ldots, v_w \in \mathcal V$ such that for any (subspace) $K \subset \mathcal V$ for which $K^\perp$ does not contain $v_i$ for any $i$ has dimension at least $h+1$.
\end{lemma}
\iffalse
\begin{proof}
We will show that $\{v_1,\ldots,v_{q^h}\}$ can be any subspace $W$ of dimension $h$ while the last vector $v_w$ is chosen orthogonal to $W$. To prove that this works we take any subspace $L$ of dimension at most $h$ and show that its orthogonal complement $L^\perp$ intersects with $W \cup v_w$. Note that $\dim L^\perp \geq \dim \mathcal V - h$. Thus, if $L^\perp$ and $W$ do not intersect we can find a set of orthogonal basis vectors for $\mathcal V$ of which $h$ span $W$ and the remaining ones span $L^\perp$. Since $v_w$ is orthogonal to $W$ it is spanned by the basis vectors in $K^\perp$ and thus $v_w \in K^\perp$ -- a contradiction.
\end{proof}
\fi 
It is now possible to prove the two main results of this section. Their proofs are self-contained and involve only random binning (\Cref{lem. rank only}) and \Cref{lemma q^l+1}:

\begin{proof}[Proof Sketch of Theorem~\ref{theo. one source}]
We use the field size $q$ that achieves the parameters $T$ and $\alpha$ in \Cref{def:floodingtime and throughput}. We furthermore choose $l$ large enough so that the decoding probability in \Cref{lem. rank only} is at most $\eps/2$. By Lemma~\ref{lem. rank only}, any node $v$ with access to the side information vector $y_v$ and $\Ceil{\frac{l}{s}(H(X|Y_v)+\delta)}$ independent equations on the blocks describing the bin index of $x$ assigned by the random binning procedure, can decode $x$ with probability at least $1 - \eps/2$. It thus remains to show that with probability $1 - \eps/2$ we have $\dim S_v \geq \Ceil{\frac{l}{s}(H(X|Y_v)+\delta)}$ after $T'$ rounds. To prove this we apply Lemma~\ref{lemma q^l+1} and get that there exists a set $Z$ of $2^{\Ceil{\frac{l}{s}(H(X|Y_v)+\delta)}}$ coefficient vectors such that if $v$ has knowledge of these vectors, it indeed has sufficiently many independent equations. Furthermore, we refer to the proof of \Cref{thm:oldthm} for the fact that knowledge of any coefficient vector (in $Z$) spreads through the network like a flooding process. As before we thus get the fact that in the assumed network model $M$ the probability that any of the coefficient vectors (in $Z$) is not known after $T+\frac{1}{\alpha}(k+(\log \epsilon^{-1}+1))$ rounds is smaller than $\epsilon/2 \cdot q^{-k}$. Setting $k=\log |Z|$ and using a union bound over all coefficient vectors in $Z$ we get as a result that indeed after $T'$ rounds the probability that $v$ has received sufficiently many independent coefficient vectors is at least $1 - \eps/2$.
\end{proof}
While the proof of Theorem~\ref{thm: two sources no sideinformation} is similar in nature to that of Theorem~\ref{theo. one source}, there is delicate point when considering multiple sources. In a single source scenario, for each terminal there is only one equation governing the rate. That is, $r \ge H(X|Y_v)+\delta$. Using Lemma~\ref{lemma q^l+1}, this rate constraint is translated into a rank constraint, namely, $dim(S_v)\ge\Ceil{\frac{l}{s}(H(X|Y_v)+\delta)}$. For more than one source, however, the rate region is given by multiple rate constraints, and one has to make sure all are satisfied. Indeed, for two sources and no side information at the nodes this can be done using a single rank constraint. For more than two sources, or when additional side information is available, a more refined analysis is required. This is the subject of the next section.
%%%%%%%%%%%%%%%%%%%%%%%%%%%%%%%%%%%%%%%%%%%%%%%%%%%%%%%%%%%%%%%%%%%%%%%%%%%%%%%%%%%%%%%%%%%%%%

\section{Characterizing Capacities in Oblivious Network Models}\label{sec:capacities}
To date, analysis of gossip schemes focused only on the dissemination time - the number of rounds required to gain the complete knowledge in the network. However, especially when dynamic networks are analyzed, it is interesting to gain a more accurate measure of the \emph{actual capacities achievable between sets of nodes}. Namely, to analyze the capacities induced by the packet exchange process in algebraic gossip. This is an interesting question in its own right, and, in particular, can give a ``black-box'' method to transfer any results of prior works that bound the rates or capacities needed between sources and sinks in the static memory-free setting to stopping times in oblivious network models. 

Herein, we develop such a characterization, and apply it to the results of \cite{Ho_et_al06} and \cite{BakshiEffros08} to obtain stopping times for gossip protocols with an arbitrary number of correlated sources and side-information, generalizing the results from the last section. 
We first introduce the required notation.

\begin{definition}
Let $T>0$, node set $V$ and active edges $E_1, E_2, \ldots, E_{T}$ be given. We define a {\em path} $P$ from $s$ to $d$ to be a sequence of nodes $P = (v_0, v_1, \ldots, v_T)$ such that $v_0 = s$, $v_T = d$ and for every $t \leq T$ we have $v_{t-1} = v_{t}$ or $(v_{t-1},v_{t}) \in E_t$. We furthermore define a set of $m$ paths $P_1, \ldots, P_M$ with weights $w_1,\ldots,w_m$ to be {\em valid} if for every $t<T$ and every $(u,v) \in E_t$ the weights of paths using $(u,v)$ sum to at most one. Lastly, we say a set of valid weighted paths {\em achieves a capacity} of $c$ between two nodes $s$ and $d$ if the weights of paths from $s$ to $d$ sum up to $c$. 
\end{definition}

Quite intuitively these paths correspond to an information flow through the network from the sources to the sink. This intuition can be made formal and one can give an explicit equivalence between the algebraic gossip protocol and random linear network coding in the classical memory-free model (e.g., \cite{Ho_et_al06}). This was done in \cite{NCoptimality} which also describes the hypergraph $G_{PNC}$ that corresponds to a sequence of active edges. We omit the details of this equivalence and instead only recall the facts needed in this paper:
% and have a very intuitive equivalence to the classical network coding setting. This equivalence was made explicit in \cite{ITW} and the following very intuitive lemma can be easily proven using the methods introduced in \cite{ITW}:
\begin{lemma}\label{lem:equivalency}
Let node set $V$, the active edges $E_1, \ldots, E_T$, destination $d \in V$ and sources $s_1, s_2, \ldots, s_k \in V$ be given. Algebraic gossip on $\{E_i\}_{i=1}^T$ is equivalent to classical random linear network coding in the transformed hypergraph $G_{PNC}$ described in~\cite{NCoptimality}. In particular, if for some integers $c_1,\ldots,c_k$, it is possible for every $s_i$ to transmit $c_i$ packets to $d$, then there exists a sequence of valid paths of weight one and a rate $c_i$ between $s_i$ and $d$. Conversely, if for some positive reals $c_1,\ldots,c_k$ there is a set of valid weighted paths that achieve a capacity $c_i$ between $s_i$ and $d$, then the capacities $c_i$ lie in the min-cut region of $G_{PNC}$.
\end{lemma}

%Given this setup the question we are interested in answering is the following:

%Suppose we sample a sequence of $T'$ sets of active edges from an oblivious network model $M$. How large does $T'$ need to be such that for a given sink $d$ and a set of $k$ source nodes $s_1, s_2, \ldots, s_k \in V$ with non-zero capacities $c_1, c_2, \ldots, c_k >0$ we can expect to find a set of valid weighted paths that achieve the capacity $c_i$ between $s_i$ and $d$ for every $i$.

%Going in this direction, we prove the following lemma:

Given this setup we show the first result in this direction:

\begin{lemma}\label{lem:disjointcapacities}
Let $M$ be a network model on a node set $V$ that floods in time $T$ with throughput $\alpha$.
For any $T'$, any $\eps >0$, any destination $d \in V$ and any set of $k$ source nodes $s_1, s_2, \ldots, s_k \in V$ with integral capacities $c_1, c_2, \ldots, c_k \geq 1$ suppose $E_1,\ldots,E_{T'}$ is a sequence of active edges on $V$ sampled from $M$. 
If $T' \geq T + \frac{1}{\alpha}(\sum_i c_i + \log \eps^{-1})$ then with probability at least $1 - \eps$ there exists a selection of valid weighted paths that achieve the capacity $c_i$ between $s_i$ and $t$ for every $i$.
%(The weights of the paths can furthermore be chosen to be one.)
\end{lemma}
\begin{proof}
We think of putting $c_i$ messages at node $s_i$ and run the standard algebraic gossip protocol for $T'$ rounds using the field size $q$ that achieves the parameters $T$ and $\alpha$ on $M$. \Cref{thm:oldthm} now shows that with probability $1 - \eps$ the sink $t$ can decode the messages. According to \Cref{lem:equivalency} this shows that there are $c_i$ mutually disjoint paths from $s_i$ to $d$ for every $i$ with weight one which achieve the desired capacities. 
\end{proof}

Note that the above lemma requires the capacities to be integral and thus essentially asks for the time until a certain number of mutually disjoint paths occur. While this is sufficient and optimal in the \emph{uncorrelated} information spreading setting this requirement can be a severe restriction.%make a drastic difference.

%{\bfseries Bernhard says: I have done changes up to here}

One setting where this makes a drastic difference is when we have $k$ sources and the total capacities needed sum up to less then one. This corresponds to asking for the time until there is a path from each of the sources to the sink -- without these paths having to be disjoint. If one considers for example the random phone call model with $n$ nodes and $k$ sources it takes in expectation $\log n + k$ time until a disjoint path between a node and each source appears while merely $\log n + \log k$ rounds are sufficient to get this for non-disjoint paths. 

The following lemma generalizes this observation and strengthens \Cref{lem:disjointcapacities} in this direction to give order optimal bounds for any set of fractional capacities:
\begin{lemma}\label{lem:fractional paths}
Let $M$ be a network model on a node set $V$ that floods in time $T$ with throughput $\alpha$.
For any $T'$, any $\eps >0$, any sink $d \in V$ and any set of $k$ source nodes $s_1, s_2, \ldots, s_k \in V$ with rates $c_1, c_2, \ldots, c_k >0$, if $T' \geq T + \frac{1}{\alpha}(\Ceil{\sum_i c_i} + \log k + \log \eps^{-1})$ then with probability at least $1 - \eps$ there exists a selection of valid weighted paths that achieve a capacity of $c_i$ between $s_i$ and $d$ for every $i$.
\end{lemma}
\begin{proof}
The idea is to combine $k$ applications of \Cref{lem:disjointcapacities} using a union bound and capacity sharing.We set the failure probability to $\eps/k$ and in the $i$th application of \Cref{lem:disjointcapacities} we set the $c_i$ to $\Ceil{\sum_i c_i}$ while setting all other capacities to zero. Out of this we get that for every $i$ with probability $1 - \eps/k$ the number of disjoint paths from $s_i$ to $d$ is at least $\Ceil{\sum_i c_i}$. Via a union bound we get that with probability of $1 - \eps$ all these paths are there. Note, that while the paths from each source are disjoint the paths starting at different sources may not be disjoint. We now take the union of these paths while choosing the weight of a paths starting at source $s_i$ to be $\frac{c_i}{\Ceil{\sum_j c_j}}$. This gives capacity of $c_i$ between $s_i$ and $d$. Furthermore, since any edge $e$ is used by at most one path going out from each source, we get that the total weight on $e$ summed over all paths is at most $\sum_{i} \frac{c_i}{\Ceil{\sum_j c_j}} \leq 1$.
\end{proof}

We will use \Cref{lem:fractional paths} to prove our main result about information dissemination with correlated data in oblivious networks. To state our result we need the following definition:
\begin{definition}[Slepian-Wolf region~\cite{SlepianWolf73}]
A capacity vector $c = (c_1,\ldots,c_k)$ is sufficient for $v \in V$ if and only if for every index subset $S \in [k]$ we have $\sum_{i \in S} c_i \geq H(X_{S} | X_{\overline{S}}, Y_v)$.
\end{definition}

Putting together \Cref{lem:fractional paths}, \Cref{lem:equivalency} and applying the results on network coding with correlated data from \cite{Ho_et_al06} and \cite{BakshiEffros08} in a black-box manner, we can now directly state our main result which generalizes and encompasses \Cref{thm: two sources no sideinformation} and \Cref{theo. one source}:
\begin{theorem}\label{thm:main-red-improved}
Suppose $M$ is an oblivious network model that allows floods in time $T$ with throughput $\alpha$.
Then, for any $\delta > 0$ and error probability $\eps > 0$, there exists an $l$ such that for any joint distribution of $X_1, \ldots, X_k$ and the $Y_{v}$'s, any packet size $s > 0$, any node $v$ and any capacity vector $(c_1,\ldots,c_k)$ that is sufficient for $v$, node $v$ will correctly decoding $x_1, \ldots, x_k$ with probability at least $1 - \eps$ after $T + \frac{1}{\alpha}(\Ceil{\frac{l}{s} \sum_{i \in [k]} c_i + \delta} + \log k + \log \eps^{-1} + \delta)$ rounds.
%Comment Bernhard: We can replace log k by \log |\{i \in [k] \ |\ c_i \neq 0\}|, i.e., the number of non-zero capacities but even so this is a bit more accurate and general it is harder to parse which is why I decided to stick with log k.
\end{theorem}

%%%%%%%%%%%%%%%%%%%%%%%%%%%%%%%%%%%%%%%%%%%%%%%%%%%%%%%%%%%%%%%%%%%%%%%%%%%%%%%%%%%%%%%%%%%%%%

\end{document}